\documentclass[11pt,reqno]{article}
\usepackage[utf8]{inputenc}
\usepackage[T1]{fontenc}
\usepackage{libertine}
\usepackage{libertinust1math}
\usepackage{inconsolata}
\usepackage{microtype}
\usepackage{algorithm}
\usepackage[noend]{algpseudocode}
\usepackage{amsmath, amssymb, amsthm}
\usepackage{bm}
\usepackage{graphicx}
\usepackage{thmtools}
\usepackage{thm-restate}
\usepackage{hyperref} % Must load hyperref before cleverref
\usepackage[capitalize, nameinlink]{cleveref}
\usepackage{comment}
\usepackage{enumerate}
\usepackage{fullpage}
\usepackage{mathtools}
\usepackage{subcaption}
\usepackage[dvipsnames]{xcolor}
\usepackage{nag}
\usepackage{authblk}
\usepackage{xspace}
\graphicspath{{figures/}{.}}

\newcommand{\declarecolor}[2]{\definecolor{#1}{RGB}{#2}\expandafter\newcommand\csname #1\endcsname[1]{\textcolor{#1}{##1}}}
\declarecolor{White}{255, 255, 255}
\declarecolor{Black}{0, 0, 0}
\declarecolor{LightGray}{216, 216, 216}
\declarecolor{Gray}{127, 127, 127}
\declarecolor{Orange}{237, 125, 49}
\declarecolor{LightOrange}{251,229, 214}
\declarecolor{Yellow}{255, 192, 0}
\declarecolor{LightYellow}{255, 242, 200}
\declarecolor{Blue}{91, 155, 213}
\declarecolor{LightBlue}{222, 235, 247}
\declarecolor{Green}{112, 173, 71}
\declarecolor{LightGreen}{226, 240, 217}
\declarecolor{Navy}{68, 114, 196}
\declarecolor{LightNavy}{218, 227, 243}

\hypersetup{
	colorlinks=true,
	pdfpagemode=UseNone,
	citecolor=Green,
	linkcolor=Navy,
	urlcolor=Navy,
}

\crefformat{equation}{(#2#1#3)}

\DeclareMathOperator*{\argmin}{arg\,min}

\newcommand{\R}{\mathbb{R}}

\newcommand{\nnz}{\textnormal{\text{nnz}}}
\newcommand{\diag}{\textnormal{\text{diag}}}

\newcommand{\var}[1]{\textnormal{\texttt{#1}}}
\newcommand{\alg}[1]{\textnormal{\textsc{#1}}}
\newcommand{\mat}[1]{\textbf{#1}}
\newcommand{\MainAlg}{\alg{FastMinimumDegree}\xspace}

\DeclarePairedDelimiter{\abs}{\lvert}{\rvert}
\DeclarePairedDelimiter{\set}{\{}{\}}
\DeclarePairedDelimiter{\parens}{(}{)}

\DeclarePairedDelimiter{\floor}{\lfloor}{\rfloor}
\DeclarePairedDelimiter{\ceil}{\lceil}{\rceil}

\theoremstyle{plain}
\newtheorem{theorem}{Theorem}[section]
\newtheorem{lemma}[theorem]{Lemma}
\newtheorem{corollary}[theorem]{Corollary}

\theoremstyle{definition}
\newtheorem{definition}[theorem]{Definition}

\pdfoutput=1

\begin{document}

\title{A Fast Minimum Degree Algorithm and Matching Lower Bound}

\author[3]{Robert Cummings\thanks{
Email: \href{mailto:robert.cummings@uwaterloo.ca}{\textsf{robert.cummings@uwaterloo.ca}}.
}}
\author[1]{Matthew Fahrbach\thanks{
Email: \href{mailto:fahrbach@google.com}{\textsf{fahrbach@google.com}}.
Part of this work was done at the Georgia Institute of Technology
while supported by an
NSF Graduate Research Fellowship under grant DGE-1650044.
}}
\author[2]{Animesh Fatehpuria\thanks{
Email: \href{mailto:afatehpuria@nuro.ai}{\textsf{afatehpuria@nuro.ai}}.
}}
\affil[1]{Google Research}
\affil[2]{Nuro}
\affil[3]{School of Computer Science, University of Waterloo}

\maketitle

\begin{abstract}
The minimum degree algorithm is one of the most widely-used heuristics for
reducing the cost of solving large sparse systems of linear equations.  It
has been studied for nearly half a century and has a rich history of bridging
techniques from data structures, graph algorithms, and scientific~computing.
In this paper, we present a simple but novel
combinatorial algorithm for computing an exact minimum degree elimination
ordering in $O(nm)$ time,
which improves on the best known time complexity of~$O(n^3)$ 
and offers practical improvements for sparse systems with small values of $m$.
Our approach leverages a careful amortized analysis,
which also allows us to derive output-sensitive bounds for the running time
of $O(\min\{m\sqrt{m^+}, \Delta m^+\} \log n)$,
where $m^+$ is the number of unique fill edges and original edges 
that the algorithm encounters
and $\Delta$ is the maximum degree of the input graph.

Furthermore, we show there cannot exist an exact
minimum degree algorithm that runs in $O(nm^{1-\varepsilon})$ time, for any
$\varepsilon > 0$, assuming the strong exponential time hypothesis.
This fine-grained reduction goes through the orthogonal vectors problem
and uses a new low-degree graph construction
called \emph{$U$-fillers}, which act as pathological inputs
and cause any minimum degree algorithm to exhibit nearly worst-case performance.
With these two results, we nearly characterize the time complexity
of computing an exact minimum degree ordering.
\end{abstract}

\newpage
\tableofcontents

\pagenumbering{gobble}
\clearpage
\pagenumbering{arabic}

\section{Introduction}
\label{sec:introduction}

The minimum degree algorithm is one of the most widely-used heuristics for
reducing the cost of solving sparse systems of linear equations.
This algorithm was first proposed by Markowitz~\cite{markowitz1957elimination}
in the context of reordering equations that arise in asymmetric
linear programming problems,
and it has since been the impetus for
using graph algorithms and data structures in
scientific computing~\cite{rose1972graph,tarjan1976graph,george1979quotient,george1981computer,george1989evolution}.
This line of work culminated in the approximate minimum degree algorithm (AMD)
of Amestoy, Davis, and Duff~\cite{amestoy1996approximate},
which has long been a workhorse in the sparse linear algebra libraries for
Julia, MATLAB, Mathematica, and SciPy.
Formally, the minimum degree algorithm is a
preprocessing step that permutes the rows and columns of a sparse
symmetric positive-definite matrix $\mat{A} \in \R^{n \times n}$,
before applying Cholesky decomposition,
in an attempt to minimize the number of nonzeros in the Cholesky factor.
Without a judicious reordering,
the decomposition typically becomes dense with \emph{fill-in}
(i.e., additional nonzeros).
The goal of the minimum degree algorithm is to efficiently compute
a permutation matrix $\mat{P}$ 
such that the Cholesky factor~$\mat{L}$ in the reordered matrix
$\mat{P}\mat{A}\mat{P}^\intercal = \mat{L}\mat{L}^\intercal$ is close to
being optimally sparse.
Finding an optimal permutation, however, is
NP-complete~\cite{yannakakis1981computing},
so practical approaches such as
minimum degree orderings,
the Cuthill–McKee algorithm~\cite{cuthill1969reducing},
and nested dissection~\cite{george1973nested} are used instead.
We direct the reader to
``The Sparse Cholesky Challenge''
in \cite[Chapter 11.1]{golub2013matrix}
to further motivate efficient reordering algorithms
and for a comprehensive survey.

The minimum degree algorithm takes advantage of a separation between the
symbolic and numerical properties of a matrix.
To see this, start by viewing the nonzero structure of $\mat{A}$ as
the adjacency matrix of an undirected graph $G$ with
$m = \nnz(\mat{A} - \diag(\mat{A}))/2$ edges.
Permuting the matrix by $\mat{P}\mat{A}\mat{P}^\intercal$ does not
change the underlying graph.
In each iteration, the algorithm
(1) selects the vertex $u$ with minimum degree,
(2) adds edges between all pairs of neighbors of $u$ to form a clique, and
(3) deletes $u$ from the graph.
Through the lens of matrix decomposition,
each greedy step corresponds to performing row and column permutations 
that minimize the number of off-diagonal nonzeros in the pivot row and column.
A clique is induced on the neighbors of $u$
in the subsequent graph because of the widely-used
\emph{no numerical cancellation assumption} (i.e., nonzero entries remain
nonzero).
Motivated by the success and ubiquity of reordering algorithms in sparse
linear algebra packages, and also by recent developments
in the hardness of computing minimum degree orderings of Fahrbach et al.~\cite{fahrbach2018graph},
we investigate the fine-grained time complexity of
the minimum degree ordering problem.

\subsection{Results and Techniques}

Our main results complement each other and nearly
characterize the time complexity to compute a minimum degree ordering.
Our first result is a new combinatorial algorithm
that outputs an exact minimum degree ordering in $O(nm)$ time.
This improves upon the best known result of $O(n^3)$ achieved by the naive algorithm.
We maintain two different data structures for the fill graph
to efficiently detect and avoid the redundant work encountered by the naive algorithm.
Using careful amortized analysis, we prove the following theorems.

\begin{restatable}[]{theorem}{TheoremAlgorithm}
\label{thm:algorithm}
The \MainAlg algorithm outputs an exact minimum degree ordering in $O(nm)$~time.
\end{restatable}

Our analysis also allows us to derive output-sensitive bounds.
The fill produced by the minimum degree heuristic is typically small in
practice, so the performance of the algorithm is often faster than $O(nm)$.
The algorithm also allows for further practical implementations,
such as using hash table-based adjacency lists.

\begin{restatable}[]{theorem}{TheoremAlgorithmOutputSensative}
\label{thm:algorithm-output-sensative}
The \MainAlg algorithm can be implemented to run in
$O(\min\set{m\sqrt{m^+}, \Delta m^+} \log n)$ time
and use $O(m^+)$ space,
where $m^+$ is the number of unique fill edges and original edges
that the algorithm encounters
and $\Delta$ is the maximum degree of the original graph.
\end{restatable}

Our second main result improves upon a recent conditional hardness theorem
of $O(m^{4/3-\varepsilon})$ for computing exact minimum degree elimination orderings
assuming the strong exponential time hypothesis~\cite{fahrbach2018graph}.

\begin{restatable}[]{theorem}{TheoremHardness}
\label{thm:hardness}
Assuming the strong exponential time hypothesis, there does not exist an
$O(m^{2-\varepsilon}\Delta^k)$ algorithm for computing
a minimum degree elimination ordering,
where $\Delta$ is the maximum degree of the original graph,
for any $\varepsilon> 0$ and $k\ge 0$.
\end{restatable}

\noindent
This result is given in its full generality above, and it implies an answer to
$O(nm^{1-\varepsilon})$-hardness conjecture posed in~\cite{fahrbach2018graph}.
Specifically, we have the following matching lower bound for our main algorithm.

\begin{restatable}[]{corollary}{CorollaryHardness}
\label{cor:hardness}
Assuming the strong exponential time hypothesis,
there does not exist an~$O(nm^{1-\varepsilon})$ time algorithm for computing
a minimum degree elimination ordering,
for any $\varepsilon > 0$.
\end{restatable}

The hardness in~\Cref{thm:hardness} also rules out the
existence of an $O(\sum_{v\in V} \deg(v)^2)$ time algorithm.
We prove our hardness results by reducing
the orthogonal vectors problem~\cite{williams2005new} to computing
a minimum degree ordering of a special graph constructed using building blocks
called \emph{$U$-fillers}.
One of our main contributions is a simple
recursive algorithm for constructing $U$-filler graphs that satisfy the challenging
sparsity and degree requirements necessary for the fine-grained reduction.
In particular, these $U$-fillers correspond to pathological sparsity patterns,
and hence adversarial linear systems,
that cause any minimum degree algorithm to
exhibit worst-case performance (i.e., to output Cholesky factors with
$\tilde\Omega(n^2)$ nonzeros).
These graphs are of independent interest and could be useful in
lower bounds for other greedy algorithms.

\subsection{Related Works}

Computing an elimination ordering that minimizes fill-in
is an NP-complete problem closely related to
chordal graph completion~\cite{yannakakis1981computing}.
Agrawal, Klein, and Ravi~\cite{agrawal1993cutting}
gave the first approximation algorithm for the minimum fill-in problem,
building on earlier heuristics by George~\cite{george1973nested}
and by Lipton, Rose, and Tarjan~\cite{lipton1979generalized}.
Natanzon, Shamir, and Sharan~\cite{natanzon2000polynomial}
later developed the first algorithm to approximate the minimum possible fill-in
to within a polynomial factor.
There has since been a wealth of recent results on
fixed-parameter tractable algorithms~\cite{kaplan1999tractability,fomin2013subexponential,cao2016chordal,bliznets2018subexponential}
and the conditional hardness of minimizing fill-in~\cite{wu2014inapproximability,bliznets2016lower,cao2017minimum,bergman2019minimum}.

Due to this computational complexity,
we rely on the practicality and efficiency of greedy heuristics.
In particular,
the multiple minimum degree algorithm (MMD) of Liu~\cite{liu1985modification}
and
the approximate minimum degree algorithm (AMD) of Amestoy, Davis, and Duff~\cite{amestoy1996approximate}
have been the mainstays for solving sparse linear systems of equations.
These algorithms, however, have some drawbacks.
MMD eliminates a maximal independent set of minimum degree vertices
in each step, but it runs in $O(n^2 m)$ time and this is known to be tight~\cite{heggernes2001computational}.
On the other hand, AMD is a single elimination algorithm
that runs in $O(nm)$ time, but achieves its speedup by
using an easy-to-compute upper bound
as a proxy to the true vertex degrees.
While many variants of the minimum degree algorithm exist,
an algorithm that computes an exact minimum degree ordering
with time complexity better than $O(n^3)$ has never been proven.
Therefore, our contributions are a significant step forward in
the theory of minimum degree algorithms.

There have also been other major advancements in the theory of minimum degree
algorithms recently.
Fahrbach et al.~\cite{fahrbach2018graph} designed an algorithm that computes a
$(1+\varepsilon)$-approximate greedy minimum degree elimination ordering
in $O(m \log^5(n) \varepsilon^{-2})$ time.
Although this result is a significant theoretical milestone,
it is currently quite far from being practical.
Ost, Schulz, and Strash~\cite{ost2020engineering} recently 
gave a comprehensive set of vertex elimination rules that
are to be used before applying a greedy reordering algorithm
and never degrade the quality of the output.
The minimum degree heuristic has also appeared in algorithms for
graph compression and coarsening~\cite{ashcraft1995compressed, chierichetti2009compressing, fahrbach2020faster}.

\section{Preliminaries}
\label{sec:preliminaries}

\subsection{Fill Graphs and Minimum Degree Orderings}

For an undirected, unweighted graph $G=(V,E)$, let $N(u) = \{v \in V :
\{u,v\} \in E\}$ denote the neighborhood of vertex~$u$ and $\deg(u) = |N(u)|$
denote its degree.
We overload the notation $N(U) = \bigcup_{u \in U} N(u)$ to be
the neighborhood of a set of vertices.
% A vertex $v \in V$ is said to be \emph{$k$-reachable} from a set of vertices $U$
% if there exists a $u \in U$ such that $\text{dist}(u,v) \le k$.
For two graphs $G_1 = (V_1, E_1)$ and $G_2 = (V_2,E_2)$,
define their union to be $G_1 \cup G_2 = (V_1 \cup V_2, E_1 \cup E_2)$.
For a given set of vertices $U$, let $K_U$
be the complete graph with vertex set $U$.

Now we introduce the idea of fill graphs.
Our notation extends that of
Gilbert, Ng and Peyton~\cite{gilbert1994efficient}.
We use the shorthand $[n]=\{1,2,\dots,n\}$ throughout the paper.

\begin{definition}
For any undirected graph $G=(V,E)$ and subset $U \subseteq V$,
the \emph{fill graph} $G^+_{U} = (V^+_U, E^+_U)$ is the
graph resulting from eliminating the vertices in $U$.
Its vertex set is $V^+_{U} = V \setminus U$, and
an edge $\{u,v\} \in E^+_{U}$ if and only if
there exists a path $(u,x_1,\dots,x_t,v)$ in $G$ such that
$x_i \in U$ for all $i \in [t]$.
\end{definition}

\noindent
Characterizing fill-in 
by $uv$-paths through eliminated vertices
allows us to compute the \emph{fill degree}
of a vertex in any partially eliminated state~$U$
without explicitly computing the eliminated matrix.
For a fill graph~$G_U^+$,
we avoid double subscripts and use the analogous notation
$\deg_{U}^+(v) = \deg_{G_{U}^+}(v)$
and $N^+_{U}(v) = N_{G_{U}^+}(v)$ to denote the
degree and neighborhood of a vertex $v \in V_{U}^+$.
Alternatively, we can use tools from linear algebra and view
$G^+_U$ as the nonzero
structure of the Schur complement of the adjacency matrix $\mat{A}(G)/U$.

An elimination ordering $p = (v_1,v_2,\dots,v_n)$
naturally induces a sequence of fill graphs $(G_0^+, G_1^+, \dots, G_n^+)$,
where $G_0^+ = G$ and $G_n^+$ is the empty graph.
Let $\deg_{i}^+(v)$ denote the degree of vertex $v \in V^+_{i}$ in the $i$-th
fill graph of this sequence.
This allows us to define a minimum degree elimination ordering.

\begin{definition}
A \emph{minimum degree elimination ordering} is permutation of the vertices 
$(v_1,v_2,\dots,v_n)$ such that
$v_i \in \argmin_{v \in V_{i-1}^+} \deg_{i-1}^+(v)$ for all $i \in [n]$.
\end{definition}

\noindent
% Several popular tie-breaking strategies exist when there are multiple vertices
% with minimum degree~\cite{george1989evolution}.
% In this paper, however, we use $(G_0^+, G_1^+, \dots, G_n^+)$ to represent the
% canonical sequence of fill graphs defined by choosing the lexicographically
% least min-degree vertex in each step.
% Finally, in our analysis we refer to the eliminated vertices as
% \emph{inactive} and the remaining vertices as \emph{active}.

\subsection{SETH-Hardness for Computing Minimum Degree Orderings}

Our lower bound for the time complexity of computing a minimum degree
elimination ordering is based on the
\emph{strong exponential time hypothesis} (SETH), which asserts that 
for every $\varepsilon > 0$, there exists an integer~$k$ such that $k$-SAT cannot be solved in $O(2^{(1-\varepsilon)n})$ time.
SETH has been tremendously useful in establishing tight conditional lower bounds for
a diverse set of problems~\cite{williams2018some}. Many of these results
rely on a fine-grained reduction to the \emph{orthogonal vectors} problem
and make use of the following theorem.

\begin{theorem}[\cite{williams2005new}]
\label{theorem:orthogonal-vectors}
Assuming SETH, for any $\varepsilon> 0$, there does not exist an $O(n^{2-\varepsilon})$
time algorithm that takes $n$ binary vectors with $\Theta(\log^2 n)$ bits
and decides if there is an orthogonal pair.
\end{theorem}
%\noindent
% We take a similar approach and build on~\cite{fahrbach2018graph},
% which shows that the orthogonal vectors problem is equivalent to
% deciding whether the fill graph of a particular partially-eliminated
% graph is a clique.
%The orthogonal vectors problem is often posed as deciding if there
%exists a pair of orthogonal vectors from two different sets, but we can reduce
%this problem to considering a single combined
%set by appending $[1, 0]^\intercal$ to all vectors
%in the first set and $[0, 1]^\intercal$ to all vectors in the second set.

%\input{algorithm}
\section{A Fast Minimum Degree Algorithm}
\label{sec:algorithm}

We present a new combinatorial algorithm called \MainAlg for
computing minimum degree orderings.
Its key feature is that it maintains the fill graph using
an implicit representation of fill-in
%(i.e., the ``generalized elements'' approach in~\cite{george1989evolution})
together with an explicit graph representation.
This combination allows us to reduce the number of redundant edge insertions to the fill graph.
Using a specialized amortized analysis, we prove the following theorems.
%We use a specialized analysis in \Cref{subsec:complexity-analysis} to prove the following theorems.
%We use a specialized amortized analysis to prove the following theorems.
%In~\Cref{subsec:complexity-analysis}, we use a specialized amortized analysis
%to bound the running time, proving the following theorems.

\TheoremAlgorithm*
\TheoremAlgorithmOutputSensative*

\subsection{The Algorithm}
We begin by describing an alternative approach for representing the
fill graphs $G_{i}^+$ as vertices are eliminated.
This \emph{hypergraph representation} stores \emph{hyperedges}
$U_1,U_2,\dots,U_k \subseteq V_{i}^+$
such that $G_{i}^{+}=K_{U_1}\cup K_{U_2}\cup \cdots \cup K_{U_k}$
at the end of each iteration.
Variants of this have frequently been used in previous literature on the
min-degree algorithm.
Our presentation closely follows that of George and
Liu~\cite{george1989evolution}, albeit with different terminology.

Given a graph $G$, we construct the initial hypergraph representation
consisting of all hyperedges $\{u,v\}$, for every edge $\{u,v\} \in E(G)$.
Next, we consider how to update the hypergraph representation
as vertices are eliminated.
Suppose we wish to eliminate a vertex $v$.
Let $U_1,U_2,\dots,U_t$ be precisely the hyperedges that contain~$v$.
Construct $W=(U_1\cup U_2\cup \cdots\cup U_t)\setminus \{v\}$.
Let $G'$ be the graph represented by the current hypergraph representation.
Then $W$ is precisely the neighborhood of $v$ in $G'$.
It follows that the fill graph obtained by eliminating $v$
is represented by removing the hyperedges $U_1,U_2,\dots,U_t$ and adding the hyperedge $W$.

% Note after each elimination, the total size $\sum_{i\in [k]}\abs{U_i}$ of the clique-set
% representation does not increase.
% This means that maintaining a clique-set representation can be maintained
% in $O(m)$ space for any number of eliminations.
% This fact is key to this technique's popularity in the previous literature.

Our algorithm to find a minimum degree ordering
can be summarized as maintaining both a hypergraph representation
and an adjacency matrix of the fill graph.
The adjacency matrix is used to efficiently compute the minimum degree vertex
in each step,
and the hypergraph representation is used to
reduce the number of redundant updates to the adjacency matrix.

When eliminating the vertex $v$ in the current fill graph $G_{i}^+$,
we find the hyperedges $U_1,U_2,\dots,U_t \subseteq V_{i}^+$ containing $v$, as described above.
We also compute $W=(U_1\cup U_2\cup \cdots \cup U_t)\setminus\{v\}=N_{i}^+(v)$.
We must add edges to the fill graph so that it contains the clique $K_W$.
Therefore, it is enough for our algorithm to try adding the edges in $K_W$
that are not in $K_{U_1}\cup K_{U_2}\cup  \cdots\cup K_{U_t}$,
since these are already guaranteed to be in $G_{i}^+$.

Below we give high-level pseudocode to describe this algorithm.
Although the implementation details of the algorithm are very important
to the efficiency of the algorithm, we defer these discussions to
\Cref{subsec:complexity-analysis}.

\begin{algorithm}[H]
\caption{A fast minimum degree algorithm for producing exact elimination orderings.}
\label{alg:fast-min-degree}

\begin{algorithmic}[1]
  \Function{\MainAlg}{adjacency list for
    undirected graph $G=(V,E)$ with $\abs{V}=n$}
  \State Initialize adjacency structure $\var{fill\_graph}$ to agree with $G$
  \For{each edge $\{u,v\} \in E$}
  \State Add $\{u,v\}$ to the list of hyperedges
  \EndFor
  \State Mark all vertices as active and initialize array $\var{elimination\_ordering}$ of size $n$
  \For{$i=1$ to $n$} 
  \State Let $a$ be the active vertex with minimum degree in $\var{fill\_graph}$
  \State Deactivate $a$ and set $\var{elimination\_ordering}[i] \gets a$
  \State Initialize $W\gets \emptyset$
  \For{each hyperedge $U$ containing $a$}
  \State Remove $U$ from the list of hyperedges
  \State Set $U\gets U\setminus\{a\}$
  \State Let $X\gets W\setminus U$ and $Y\gets U\setminus W$
  %\State Let $J\gets U\setminus W$
  \For{each pair $(x,y)$ in $X\times Y$}
  \State Add edge $\{x,y\}$ to $\var{fill\_graph}$ if not present
  \EndFor
  \For{each vertex $b \in Y$}
%  \State Add $b$ to $W$
  \State Remove edge $\{a,b\}$ from $\var{fill\_graph}$
  \EndFor
  \State Update $W \gets W \cup Y$
  \EndFor
  \State Add $W$ to the list of hyperedges
  \EndFor
  \State \textbf{return} $\var{elimination\_ordering}$
  \EndFunction
\end{algorithmic}
\end{algorithm}

We claim that the algorithm is correct and maintains the desired state at the
end of each iteration.
%The proof is straightforward and deferred to \Cref{app:algorithm}.
%\todo{Update all the ``proof is deferred'' sentences.}

\begin{restatable}[]{lemma}{LemmaCorrectness}
  \label{lem:correctness}
  \alg{FastMinimumDegree} produces a correct minimum degree ordering $\var{elimination\_ordering}$.
  Furthermore, suppose  $(G_0^+,G_1^+,\dots,G_n^+)$ is the sequence of fill graphs
  induced by $\var{elimination\_ordering}$.
  Then, at the end of each iteration $i \in [n]$ of the \alg{FastMinimumDegree} algorithm:
  \begin{enumerate}
  \item The state of $\var{fill\_graph}$ corresponds to the fill graph $G_{i}^+$.
  \item The list of hyperedges $U_1,U_2,\dots,U_k$
    satisfies $K_{U_1}\cup K_{U_2}\cup \cdots \cup K_{U_k}=G_i^+$
      (i.e., the hypergraph representation is maintained).
  \end{enumerate}
\end{restatable}

\begin{proof}
We proceed by induction and show that these invariants hold at the end
of each iteration.
Before any vertices are eliminated (line~5)
both properties are true since $G_{0}^+ = G$.
Assume the claim as the induction hypothesis
and suppose that the algorithm begins iteration $i \in [n]$.
Clearly \MainAlg selects a minimum degree vertex in $G_{i-1}^+$
and updates $\var{elimination\_ordering}$ correctly.

Now we consider what happens when vertex $a$ is eliminated.
Let $U_1, U_2, \dots, U_t \subseteq V$ be all the hyperedges 
containing $a$ at the start of iteration $i$.
When the algorithm reaches line~19,
we have $W = (U_1 \cup U_2 \cup \dots \cup U_t) \setminus\set{a}$.
All hyperedges $U_1,U_2,\dots,U_t$ are removed from the hypergraph
representation by the end of this iteration,
and $W$ is added in their place.
Since $W = N_{i-1}^+(a)$ by the induction hypothesis,
it follows that the new list of hyperedges at the end of iteration $i$
corresponds to $G_{i}^+$ and satisfies the second property.

Ensuring that the state of $\var{fill\_graph}$ is updated
correctly requires a little more work.
By the induction hypothesis, $\var{fill\_graph}$ is equal to
$G_{i-1}^+$ at the beginning of iteration $i$.
The hypergraph representation also corresponds to $G_{i-1}^+$ at
the start of this iteration,
so it follows that for each hyperedge~$U$ that contains $a$,
the edges in $K_{U}$ are present in $\var{fill\_graph}$.
Lines 9--18 attempt to insert edges from $K_{W}$ into
$\var{fill\_graph}$ that may be missing.
In particular,
clique $K_W$ is constructed one hyperedge at a time to
reduce redundant work.
When processing the $j$-th hyperedge $U_{j}$,
we have $X = (U_1 \cup U_2 \cup \dots U_{j-1}) \setminus U_{j}$
and
$Y = U_j \setminus (U_1 \cup U_2 \cup \dots \cup U_{j-1})$
on line~13.
Only edges $\{x,y\} \in X\times Y$
spanning the symmetric difference can be
missing from $\var{fill\_graph}$,
so at the end of iteration $i$,
all edges in $K_W$ have been considered.
Finally, all edges adjacent to $a$ in $G_{i-1}^+$
are removed in lines 16--17.
Therefore, it follows that
the state of $\var{fill\_graph}$
corresponds to $G_{i}^+$ at the end of the iteration.
This proves the first property and completes the proof by induction.
\end{proof}

\subsection{Complexity Analysis}
\label{subsec:complexity-analysis}

To thoroughly investigate the time complexity of minimum degree algorithms, we
first relax any rigid space requirements~\cite{george1980minimal,heggernes2001computational}.
This allows us to more conveniently present an algorithm whose running time
matches the SETH-based lower bound in~\Cref{cor:hardness}.
%We also note that all the missing proofs in this subsection are deferred to
%\Cref{app:algorithm} in order to better demonstrate how our analysis comes
%together.

Our goal is to bound the running time of \MainAlg not just by $O(nm)$, but by more
accurate bounds in terms of the total fill produced by the returned minimum
degree ordering.
Since the fill produced by the minimum degree heuristic is typically small in practice,
our analysis shows that the performance of the algorithm will often be better
than $O(nm)$.  Below are the main quantities used to describe the total fill.

\begin{definition}
  Let $E^+=\bigcup_{i=0}^{n} E_{i}^+$ be the set of all edges in the
  fill graph at some iteration, and let $m^+=\abs{E^+}$.
\end{definition}

We have that $E^+$ is precisely the edge set of the input graph, together
with all edges inserted throughout the course of the algorithm.
Clearly $(V,E^+)$ is a simple graph, and so $m^+\le n^2$.
These edges are important because they correspond to nonzero entries
in the corresponding Cholesky factor $\mat{L}$.
In fact, $m^+$ is precisely the number of nonzeros.
Intuitively, $m^+$ is the size of the ``output'' of the reordering procedure,
and is the quantity that the minimum degree heuristic is designed to minimize.

First, we claim that any adjacency structure for the evolving fill graph can
easily be extended to handle fast minimum degree queries.
This technique uses a bucket queue 
and is inspired by ideas in~\cite{matula1983smallest}.

\begin{restatable}[]{lemma}{LemmaMinDegreeQuery}
  \label{lem:min-degree-query}
  The $\var{fill\_graph}$ adjacency structure
  can be augmented to support
  the selection of the minimum degree vertex,
  increasing the total running time by at most $O(m^+)$.
\end{restatable}

\begin{proof}
  We maintain the data structure of $n$ doubly linked lists,
  which are ``buckets'' corresponding to degrees.
  Each vertex in the graph has a node in exactly one bucket.
  Every time we update the degree of a vertex,
  we delete it from one bucket and add it to another.
  This operation can be performed in $O(1)$ time.
  
  The total number of degree updates is also $O(m^+)$.
  For each entry of the adjacency matrix,
  there are at most two degree updates:
  one when the edge is added (line 15),
  and one when it is deleted (line 17).
  These operations are only applied to precisely
  the $O(m^+)$ entries corresponding to the edges in $E^+$.

  Selecting the minimum degree vertex $v$ is done by checking the buckets
  in order of increasing degree, until a nonempty bucket is found.
  The running time for the $i$-th step is linear in $\deg^+_i(v)$,
  which is at most the degree of $v$ in the graph $(V,E^+)$.
  So in total, the running time is $O(m^+)$.
\end{proof}

The next lemma builds on \Cref{lem:min-degree-query}
to bound the running time of the algorithm
by the total fill-in, the number of edge insertion attempts,
and the cost of edge updates.
The proof uses a simple credit analysis to show that computing all of the sets
$X$, $Y$, and $W$ over the course of the algorithm is not too expensive.

\begin{restatable}[]{lemma}{LemmaCliqueSet}
  \label{lem:clique-set}
  Let $k$ be the total number of edge insertion attempts,
  and let $c$ be the cost of querying, inserting, or deleting an edge.
  Then the algorithm can be implemented to run in $O(c(m^++k))$ time.
\end{restatable}

\begin{proof}
  Let $(v_1,v_2,\dots,v_n)$ be the ordering found by the algorithm.
  The number of edge removals performed in iteration $i$ is $\deg^+_{i-1}(v_i)$.
  Summed across all iterations, the total number of edge removals is $O(m^+)$.
  The other queries and insertions to the adjacency structure happen
  at most $k$ times.
  Therefore, the total running time of adjacency structure operations is $O(c(m^++k))$.

  We claim that the rest of the algorithm can be implemented to run in a total
  of $O(m^++k)$ time.
  The first step of this is an amortized analysis that assigns $\abs{U}$ credits
  to each hyperedge $U$ when it is created, so that when the hyperedge $U$ is removed,
  the algorithm can afford an extra running time of $O(\abs{U})$.
  The initial hypergraph requires a total of $O(m)$ credits.
  
  Consider iteration $i$ of the algorithm,
  and let $U_1,U_2,\dots,U_t$ be the hyperedges containing $v_i$.
  The set $W$ produced by the algorithm is, by the end of the iteration,
  equal to $(U_1\cup U_2\cup \cdots\cup U_t)\setminus \set{v_i}=N_{i}^+(v_i)$.
  So, the hyperedge added in iteration $i$ has size $\deg^+_i(v_i)$,
  which summed over all iterations, requires $O(m^+)$ credits.
  This amortized analysis assigns $\sum_{j=1}^t\abs{U_j}$ credits
  to the iteration where $U_1,U_2,\dots,U_t$ are the hyperedges
  being manipulated and removed.
  
  Therefore, it suffices to show that iteration $i$ runs in $O(\sum_{j=1}^t\abs{U_j}+k_i)$
  time, where $k_i$ is the number of edge insertion attempts in iteration $i$.
  First, consider the loop over the hyperedges $U_1,U_2,\dots,U_t$.
  For $j\in [t]$, let $W_j$ be the set $W$ at the beginning of the iteration
  that considers $U_j$.
  Most of the operations in the loop clearly run in $O(\abs{U_j})$ time.
  
  The exception is the computation of $X_j=W_j\setminus U_j$ and $Y_j=U_j\setminus W_j$.
  In general, finding $A\setminus B$ when $A$ and $B$ have values in $[n]$
  is done by having a global array of size $n$
  modified to represent the contents of $B$.
  Then $A\setminus B$ can be found in $O(\abs{A})$ time.
  This also requires $O(\abs{B})$ time to modify and reset the relevant array entries.
  Our implementation maintains such an array that represents $W$
  as it changes over all $t$ iterations, in a total of $\abs{N_i^+(v_i)}$ time.
  This allows $Y_j$ to be computed in $O(\abs{U_i})$ time.

  If $Y_j$ is empty, our algorithm can end iteration $j$ of the inner loop early.
  Otherwise, the number of edge insertion attempts in this iteration
  is at least $\abs{X_j\times Y_j}\ge \abs{X_j}$.
  In this case, the set $X_j$ is found in $O(\abs{W_j}+\abs{U_j})$ time.
  Note that
  \[
    \abs{W_j}+\abs{U_j}\le (\abs{W_j\setminus U_j}+\abs{U_j})+\abs{U_j}=O(\abs{X_j}+\abs{U_j}).
  \]
  So, the running time of these operations summed over all iterations $j\in [t]$ is $O(\sum_{j=1}^t\abs{U_j}+k_i)$.

  Finally, we consider the hypergraph data structure used
  to efficiently determine the hyperedges that contain a vertex $a$,
  and how this data structure is updated throughout the algorithm.
  In the implementation, each hyperedge ever created
  has a marker that determines whether it is valid or invalid
  (i.e., whether the hyperedge has since been removed).
  There is also a array of lists that determines,
  for each vertex $a$,
  the list of all hyperedges (valid or invalid) that contain it.
  This data structure is efficient to use and maintain.
  When a hyperedge $U$ is inserted,
  we iterate through all $b\in U$ and add a pointer to $U$
  to the list of vertex $b$.
\end{proof}

The bottleneck of the algorithm is the total number of attempted
edge insertions across all vertex pairs $(x,y)\in X\times Y$.
A trivial bound on the total running time of this step is $O(n^3)$.
Improving on this requires a careful amortized analysis.
Note that the algorithm may attempt to insert
an edge between the same pair of vertices multiple times.
Our approach is to investigate and bound
precisely how many times an edge in $E^+$ is attempted to be inserted.
The lemma below is the key technical result in the analysis of the algorithm.

\begin{lemma}
  \label{lem:sum-min-degree-bound}
  The number of edge insertion attempts is at most $\sum_{\set{u,v}\in E^+}\min(\deg(u),\deg(v))$,
  where the $\deg(v)$ is the degree of $v$ in the original graph.
\end{lemma}

\begin{proof}
  Each $\set{u,v}$ not in $E^+$ is never inserted by the algorithm.
  So, it suffices to show that for all $\set{u,v}\in E^+$,
  the insertion of $\set{u,v}$ is attempted at most $\min(\deg(u),\deg(v))$ times.
  Suppose $\set{u,v}\in E^+$, and assume without loss of generality
  that $\deg(u)=\min(\deg(u),\deg(v))$.
  Let $f(u,v)$ be the number of hyperedges that contain $u$ but not $v$.
  The quantity changes over time as the algorithm progresses.
  
  First, we claim that until $u$ or $v$ is eliminated,
  $f(u,v)$ only decreases over time.
  This is because the only operations done to hyperedges
  is to merge them and to delete eliminated vertices.
  Merging hyperedges does not increase the number that contain $u$ but not $v$,
  so until $u$ or $v$ is eliminated, $f(u,v)$ cannot increase.

  Next, we claim that for each time the algorithm attempts to insert $\set{u,v}$,
  $f(u,v)$ strictly decreases.
  This is because an insertion attempt only happens when a set containing $u$
  but not $v$ is merged with a set that contains $v$.
  So, after the merge operation, $f(u,v)$ has decreased by at least one.

  Finally, note that when all of the hyperedges are first initialized,
  $f(u,v)$ is at most $\deg(u)$.
  All insertion attempts for $\set{u,v}$ must happen before $u$ or $v$ is eliminated,
  and during this time, $f(u,v)\ge 0$.
  Thus, there are at most $\deg(u)=\min(\deg(u),\deg(v))$ insertion attempts
  for the edge $\set{u,v} \in E^+$.
\end{proof}

\begin{corollary}
  \label{cor:constant-degree}
  There are $O(\Delta m^+)$ edge insertion attempts,
  where $\Delta$ is the maximum degree of the input graph.
\end{corollary}

We note that the approximate minimum degree algorithm (AMD)
of Amestoy, Davis, and Duff~\cite{amestoy1996approximate}
also emits a tighter time complexity of $O(m^+)$
on bounded-degree graphs~\cite{heggernes2001computational}.
This is one of the key reasons why it is exceptionally useful for reordering
linear systems of equations on grid graphs.

Now we present a useful folklore result that allows us to make the sum over
minimum degrees expression
in \Cref{lem:sum-min-degree-bound} more comprehensible.
This inequality leverages the existence of a structured edge orientation of its
input graph
and holds for any nonnegative assignment to the vertices,
including the degree function.

\begin{restatable}[]{fact}{FactVertexFunction}
  \label{fact:vertex-function}
  Let $G=(V,E)$ be a simple graph and $m=\abs{E}$.
  For any vertex function $f : V\to \R_{\ge 0}$, we have
  \[
    \sum_{\set{u,v}\in E}\min\parens*{f(u),f(v)}
    \le \sqrt{2m}\sum_{v\in V} f(v).
  \]
\end{restatable}

\begin{proof}
  We first claim there is an orientation of $E$ such that
  each vertex has out-degree at most $\sqrt{2m}$.
  Orient the edge $\set{u,v}$ such that if $u\rightarrow v$, then $\deg(u)\le \deg(v)$.
  Let $\deg_{\text{out}}(u)$ denote the out-degree of $u$.
  Then
  \[
    2m\ge \sum_{v \in N_{\text{out}}(u)}\deg(v)
    \ge \sum_{v \in N_{\text{out}}(u)}\deg(u)
    \ge \deg_{\text{out}}(u)^2.
  \]
  Therefore, every vertex has out-degree at most $\sqrt{2m}$.
  The result then follows from the inequality
  \[
    \sum_{\set{u,v}\in E}\min(f(u),f(v))
    =\sum_{u \in V} \sum_{v \in N_{\text{out}}(u)}\min(f(u),f(v))
    \le\sum_{u \in V} \sum_{v \in N_{\text{out}}(u)} f(u)
    \le \sqrt{2m} \sum_{u\in V} f(u).
  \]
  which completes the proof.
\end{proof}

\begin{restatable}[]{corollary}{CorCleanInsertionAttempts}
\label{cor:clean-insertion-attempts}
  The number of edge insertion attempts is $O(m\sqrt{m^+})$.
  In particular, it is at most $O(nm)$.
\end{restatable}

\begin{proof}
  First, we apply \Cref{fact:vertex-function}
  to the graph $(V,E^+)$ using the degree function
  of the input graph $G=(V,E)$.
  The $O(m \sqrt{m^+})$ bound then follows from \Cref{lem:sum-min-degree-bound}
  and the fact that $\sum_{v\in V}\deg(v)=2m$.
  The $O(nm)$ bound is a consequence of the inequality $m^+ \le n^2$.
\end{proof}

We are now prepared to prove our main theorems about the \MainAlg algorithm.
Our first result states that the algorithm runs in $O(nm)$ time, matching the
conditional hardness result in~\Cref{cor:hardness}.
Our second result is an output-sensitive bound on the running time that
demonstrates the nuances in the speed of the algorithm across various inputs.
The two proofs are analogous to each other,
but use different adjacency structures for the fill graph
in order to achieve a space-time tradeoff that is beneficial in practice.

\begin{proof}[Proof of \Cref{thm:algorithm}]
  We use an adjacency matrix to represent $\var{fill\_graph}$.
  Thus, initializing $\var{fill\_graph}$ requires $O(n^2)$ time and space.
  Note that we may assume $G$ is connected, so $m \ge n-1$.
  Using \Cref{lem:clique-set}, the algorithm runs in $O(m^+ + k)$ time
  since adjacency matrices support constant-time edge queries, insertions,
  and deletions.
  Here $k$ is the total number of edge insertion attempts.
  By~\Cref{cor:clean-insertion-attempts}, we have $k = O(nm)$.
  Therefore, the overall running time is bounded by
  $O(n^2 + m^+ + k) = O(n^2 + n^2 + nm) = O(nm)$, as desired.
\end{proof}

\begin{proof}[Proof of \Cref{thm:algorithm-output-sensative}]
  We represent $\var{fill\_graph}$ as an adjacency list where
  the neighborhood of every node is stored in a balanced binary search tree
  \cite[Chapter 13]{cormen2009introduction}.
  Each edge query and update costs $O(\log n)$.
  Initializing $G$ requires $O(m \log n)$ time and $O(m)$ space.
  The algorithm then runs in $O((m^+ + k) \log n)$ time
  by \Cref{lem:clique-set}, where $k$ is the number of edge insertion attempts.
  The output-sensitive bounds in \Cref{cor:constant-degree}
  and \Cref{cor:clean-insertion-attempts} imply that
  $k = O(\min\set{m\sqrt{m^+}, \Delta m^+})$.
  Since we have $m^+ = \sqrt{m^+} \sqrt{m^+} \le n \sqrt{m^+} = O(m\sqrt{m^+})$,
  the time complexity follows.
  Further, the total space used is $O(m^+)$
  because of the binary search trees.
\end{proof}

\section{Improved SETH-Based Hardness}
\label{sec:hardness}
Now we affirmatively answer a conjecture of
Fahrbach et al.~in~\cite{fahrbach2018graph}
about the conditional hardness of finding
a minimum degree ordering.
In particular, we prove the following stronger result.

\TheoremHardness*

The previous best SETH-based lower bound in~\cite{fahrbach2018graph}
ruled out the existence of a nearly linear time algorithm for computing
an exact minimum degree ordering by showing that a $O(m^{4/3-\varepsilon})$
time algorithm could be used to solve any instance
of the orthogonal vectors problem in subquadratic time,
for any $\varepsilon > 0$.
Our approach has several similarities to that of Fahrbach et. al,
and a consequence of our main hardness result
gives a nearly matching lower bound for the running time of the \MainAlg algorithm.

\CorollaryHardness*

The key to our reduction is a recursive algorithm
for constructing a graph with $O(n\log n)$ vertices and edges such that
after any minimum degree ordering
eliminates all but $n$ vertices, the resulting graph is $K_n$.
Although this construction was originally motivated by connections to
SETH-based hardness, it has several interesting standalone properties.
In particular, it demonstrates a case where the minimum degree heuristic has
extremely poor performance
(i.e.,
it is an input of size $O(n \log n)$ that always results in $\Omega(n^2)$ fill edges).

\subsection{Constructing Sparse Min-Degree $U$-Fillers}
Our goal in this subsection is to construct sparse graphs that contain the
vertex set $U$ and have the additional property that by repeatedly eliminating
a minimum degree vertex, the resulting fill graph is eventually $K_U$.
We begin by defining several specific properties that are helpful for
presenting our construction.

\begin{definition}
\label{def:U-filler}
  A \emph{$U$-filler} is a graph $G$ with vertex set
  $U \cup W$, where $U \cap W = \emptyset$, such that
  after eliminating all of the vertices in $W$,
  the resulting fill graph is $G_{W}^+ = K_U$.
  We call $W$ the set of \emph{extra vertices}.
\end{definition}

\noindent
Many of the graphs we construct have a subset of extra vertices
similar to the one in this definition.
When taking the union of graphs with extra vertices,
we always assume the sets of extra vertices are disjoint.

A simple example of a $U$-filler is the star graph with vertices $U \cup \{w\}$,
where $w$ is the center vertex with degree $|U|$.
Although stars are sparse, their maximum degree can be arbitrarily large.
This is the main challenge in designing adversarial inputs for minimum
degree algorithms.
Note that $K_U$ is itself a $U$-filler.
%Another trivial example is $K_U$ itself (i.e., $W = \emptyset$).

\begin{definition}
\label{def:U-filler-min-degree}
  A $U$-filler is \emph{min-degree} if after eliminating any proper
  subset of extra vertices $X \subset W$,
  all of the minimum-degree vertices in $G_{X}^+$ are extra vertices.
\end{definition}

\noindent
It immediately follows from \Cref{def:U-filler-min-degree}
that every minimum degree elimination ordering
of a min-degree $U$-filler
eliminates all of the extra vertices in $W$
before any of the vertices in $U$.

Another key property we use in our construction is the following
notion of degree bound.
\begin{definition}
  A $U$-filler is \emph{$d$-bounded}
  if after eliminating any subset of extra vertices $X \subseteq W$,
  all remaining extra vertices have degree at most $d$ in $G_{X}^+$.
\end{definition}

The first graph construction we use is called a \emph{$U$-comb}.
We define a $U$-comb to consist of a path of~$\abs{U}$ extra vertices
together with $\abs{U}$ edges that form a matching between
the path vertices and those in $U$.
\begin{lemma}
  \label{lem:comb}
  A $U$-comb is a $\abs{U}$-bounded $U$-filler.
\end{lemma}
\begin{proof}
  The extra vertices of the $U$-comb form a connected component,
  so eliminating them results in $K_U$.
  Now suppose some subset of extra vertices is eliminated,
  and let $v$ be one of the remaining extra vertices.
  Consider the neighbors of $v$ in the fill graph.
  Let $w\ne v$ be another extra vertex of the $U$-comb.
  If $w$ has not been eliminated,
  $v$ cannot be adjacent to $w$'s neighbor in $U$.
  So for any of the $\abs{U}-1$ other extra vertices,
  $v$ has at most one neighbor in the fill graph.
  Since $v$ is adjacent to a single other vertex in $U$,
  it follows that the $U$-comb is $\abs{U}$-bounded.
\end{proof}

There is also a straightforward way to combine $U$-combs
to obtain $d$-bounded $U$-fillers for any $d\ge 2$,
although the size of the solution depends on the ratio $\abs{U}/d$.
\begin{lemma}
\label{lem:bounded}
  Let $d\ge 2$ and suppose $\abs{U}/d\le c$.
  Then we can construct a $d$-bounded $U$-filler with $O(\abs{U}c)$ edges
  and maximum degree $O(c)$.
\end{lemma}

\begin{proof}
  We partition $U$ into $O(c)$ parts $U_1,U_2,\dots,U_k$ each of size at most $d/2$.
  We then let $G$ be the union of $(U_i\cup U_j)$-combs
  over all unordered pairs $\{i,j\} \in \binom{[k]}{2}$.
  The result of eliminating all extra vertices of $G$ is the union of $K_{U_i\cup U_j}$,
  which is exactly $K_U$.
  Since $\abs{U_i\cup U_j}\le d/2+d/2=d$, the combs are all $d$-bounded by \Cref{lem:comb}.
  Therefore, $G$ is a $d$-bounded $U$-filler.
  We constructed $G$ from $O(c^2)$ combs each with $O(d)$ edges,
  so $G$ has $O(c^2d)=O(\abs{U}c)$ edges.
  The extra vertices of the combs have maximum degree $3$,
  but each vertex in $U$ has degree exactly $k-1$,
  so the maximum degree of $G$ is $O(c)$.
\end{proof}

We now use \Cref{lem:bounded} to construct min-degree $U$-fillers.
The main idea of our approach is to (1) recursively construct
min-degree fillers for two halves of $U$, and (2) connect the halves using
a new $U$-filler with~$O(|U|)$ edges
whose extra vertices are guaranteed to be eliminated before any vertices in $U$.
Combining this idea with divide-and-conquer,
we show that the solution is of size $O(\abs{U} \log\abs{U})$.
Before proceeding to the proof, we reiterate that all extra vertices introduced
in this construction are unique.
%Note that every extra vertex in the construction below is unique.

%However, the only purpose of these fillers is to
%increase the degree of vertices in $U$ enough to allow
%using a $U$-filler of size $O(\abs{U})$.
%By applying this idea with divide-and-conquer,
%the total size of the solution is $O(\abs{U}\log\abs{U})$.

\begin{theorem}
  \label{thm:mindegreefiller}
  We can construct a $(\abs{U}-3)$-bounded min-degree $U$-filler
  with $O(\abs{U}\log\abs{U})$ vertices and edges, and maximum
  degree $O(\log\abs{U})$.
\end{theorem}
\begin{proof}
  We prove this theorem by describing
  a recursive divide-and-conquer algorithm.
  In the base case when
  $\abs{U}\le 7$, simply return $K_U$.
  Since $K_U$ has no extra vertices,
  it is a $(\abs{U}-3)$-bounded min-degree $U$-filler.
  
  Now suppose $\abs{U}\ge 8$, and partition $U$ into $U_1$ and $U_2$
  with sizes $\floor{\abs{U}/2}$ and $\ceil{\abs{U}/2}$,
  respectively.
  We then recursively apply the theorem to construct
  min-degree fillers $G_1$ and $G_2$ for $U_1$ and $U_2$.
  Finally, we apply \Cref{lem:bounded} to construct a
  $(\floor{\abs{U}/2}-2)$-bounded $U$-filler $G_3$.
  (We can do this because $\floor{\abs{U}/2}-2\ge 8/2-2=2$.)
  Then we return the union $G_1\cup G_2\cup G_3$.

  Since $\ceil{\abs{U}/2}-3\le \floor{\abs{U}/2}-2$
  we have that $G_1$, $G_2$, and $G_3$ are $(\floor{\abs{U}/2}-2)$-bounded.
  Since $G_3$ is a $U$-filler, so is $G_1\cup G_2\cup G_3$.
  Therefore, $G_1\cup G_2\cup G_3$ is a $(\floor{\abs{U}/2}-2)$-bounded $U$-filler.
  Since $\abs{U}\ge 8$, this implies in particular that $G_1\cup G_2\cup G_3$
  is $(\abs{U}-3)$-bounded.

  Suppose for contradiction that $G_1\cup G_2\cup G_3$ is not min-degree,
  i.e., after eliminating some proper subset of extra vertices,
  there is a vertex $u\in U$ of minimum degree.
  Let $U_i$ be the part of $U$ that contains $u$.
  Note that the degree of $u$ in $G_1\cup G_2\cup G_3$
  after elimination
  is at least that of $u$ in $G_i$ after eliminating
  the corresponding extra vertices of $G_i$.
  So for $u$ to have minimum degree,
  the min-degree property of $G_i$ implies
  all extra vertices of $G_i$ are eliminated.
  Then, since $G_i$ is a $U_i$-filler,
  the degree of $u$ is at least $\abs{U_i}-1\ge \floor{\abs{U}/2}-1$.
  But since $G_1\cup G_2\cup G_3$ is $(\floor{\abs{U}/2}-2)$-bounded,
  this implies that all extra vertices were eliminated, which is impossible.
  Therefore, $G_1\cup G_2\cup G_3$ is min-degree.

  Finally, we claim that $G_1\cup G_2\cup G_3$ has $O(\abs{U}\log\abs{U})$ edges
  and maximum degree $O(\log\abs{U})$.
  Since 
  \[
    \abs{U}/(\floor{\abs{U}/2}-2)\le \abs{U}/(\abs{U}/2-3)=2/(1-6/\abs{U})\le 8,
  \]
  $G_3$ has $O(\abs{U})$ edges and maximum degree $O(1)$
  by \Cref{lem:bounded}.
  Now consider the divide-and-conquer nature of the construction.
  The number of edges when $\abs{U}=n$ follows a recurrence of the form $f(n)=2f(n/2)+O(n)$,
  which has the solution $f(n)=O(n\log n)$.
  So, the resulting graph has $O(\abs{U}\log\abs{U})$ vertices and edges.
  The degrees of all extra vertices are $O(1)$,
  and any vertex in $U$ is adjacent to at most $O(\log \abs{U})$
  fillers across all recursive levels.
  Therefore, the overall maximum degree is $O(\log \abs{U})$.
\end{proof}

\subsection{Reduction from Orthogonal Vectors}

We now use our min-degree $U$-filler construction to demonstrate the hardness of
finding an exact minimum degree ordering, assuming SETH.
Our connection to SETH is through the following problem.
\begin{definition}
  The \emph{clique union problem} takes as input
  a set of vertices $V$ with $\abs{V}=n$ and subsets of the vertices
  $U_1,U_2,\dots,U_d \subseteq V$, where $d=\Theta(\log^2n)$,
  and asks whether $K_{U_1}\cup K_{U_2} \cup \dots \cup K_{U_d}=K_V$.
\end{definition}

\begin{lemma}
  \label{lem:clique-union-problem}
  Assuming the strong exponential time hypothesis,
  for any $\varepsilon>0$, there is no $O(n^{2-\varepsilon})$ time algorithm
  for the clique union problem.
\end{lemma}

\begin{proof}
  Let each set $U_i$ correspond to the set of vectors with
  a nonzero entry in the $i$-th dimension.
  There is a pair of orthogonal vectors if and only if
  the union $K_{U_1}\cup K_{U_2} \cup \dots \cup K_{U_d}$ is not the complete graph $K_V$.
  Therefore, the result follows from \Cref{theorem:orthogonal-vectors}
  (\cite{williams2005new}).
\end{proof}

Our approach to prove \Cref{thm:hardness} using \Cref{lem:clique-union-problem}
is outlined in \Cref{alg:clique-union} below.
It is essentially a reduction from the clique union problem to the problem
of finding an exact minimum degree ordering.
The graph $G$ on which we call the minimum degree ordering subroutine
is built as a union over $U_i$-fillers produced by \Cref{thm:mindegreefiller}.
As we prove below, $G$ has special properties that,
given any minimum degree ordering,
allow us to efficiently determine the answer to the clique union instance.

\begin{algorithm}[H]
\caption{Decides if $K_{U_1}\cup K_{U_2}\cup\dots\cup K_{U_d}=K_V$.}
\label{alg:clique-union}
\begin{algorithmic}[1]
  \Function{\alg{CliqueUnion}}{$\text{vertex sets } V \text{ and } U_1,U_2,\dots,U_d \subseteq V$}
  \For{$i=1$ to $d$}
  \State Let $G_i$ be a min-degree $U_i$-filler constructed using \Cref{thm:mindegreefiller}
  \EndFor
  \State Let $G\gets G_1\cup G_2\cup\cdots\cup G_d$
  \State Let $W$ be the set of extra vertices of $G$
  \State Set $\var{elimination\_ordering}\gets \alg{MinimumDegreeOrdering}(G)$
  \For{$i=1$ to $\abs{W}$}
  \If{$\var{elimination\_ordering}[i]\notin W$}
  \State \textbf{return} $\var{false}$ \label{line:early-termination}
  \EndIf
  \EndFor
  \State Set $v\gets \var{elimination\_ordering}[\abs{W}+1]$ \label{line:choose-vertex}
  \State Determine the vertices in $G$ that are reachable from $v$
  via paths whose internal vertices are in $W$ \label{line:reachability}
  \State Let $k$ be the number of vertices in $V$ that are reachable from
  $v$, including $v$
  \State \textbf{return} $k=\abs{V}$
  \EndFunction
\end{algorithmic}
\end{algorithm}
\begin{lemma}
  \label{lem:reduction-correctness}
  If $\alg{MinimumDegreeOrdering}$ returns a minimum degree ordering,
  then $\alg{CliqueUnion}$ correctly decides
  if $K_{U_1}\cup K_{U_2}\cup\dots\cup K_{U_d}=K_V$.
\end{lemma}

\begin{proof}
  First, consider the case where \alg{CliqueUnion} terminates
  early by returning $\var{false}$ on
  %TODO: Check line number.
  line~9.
  Then for some proper subset $X\subset W$,
  a minimum degree vertex $v$ of $G^+_X$ is not in $W$.
  Suppose $i$ is an index such that $v\in U_i$.
%  (If no such $i$ exists, the instance is trivially decidable.)
  By \Cref{thm:mindegreefiller}, $G_i$ is a  min-degree $U_i$-filler.
  By considering the fill graph of $G_i$ after eliminating
  all vertices in $X\cap V(G_i)$, it follows that all the extra vertices of $G_i$ have been eliminated.
  Therefore, the degree of $v$ in $G^+_X$ equals that of $v$ in $K_{U_1}\cup K_{U_2}\cup\cdots\cup K_{U_d}$.
  But \Cref{thm:mindegreefiller} guarantees that for all $i\in [d]$,
  $G_i$ is $(\abs{U_i}-3)$-bounded,
  and thus $(\abs{V}-2)$-bounded.
  So since $X\ne W$, the degree of $v$ in $K_{U_1}\cup K_{U_2}\cup\cdots\cup K_{U_d}$ is at most $\abs{V}-2$.
  Therefore, the algorithm correctly decides that $K_{U_1}\cup K_{U_2}\cup\cdots\cup K_{U_d} \ne K_V$.
  
  Now assume that the algorithm does not terminate early.
  %TODO: Check line number.
  Then the vertex~$v$ that is chosen on line 10
  is the minimum degree
  vertex of $G^+_W$.
  Since $G_i$ is a $U_i$-filler for all $i\in [d]$, 
  $G^+_W=K_{U_1}\cup K_{U_2}\cup\dots\cup K_{U_d}$.
  Recall that an edge $\{u,v\}$ is in $G^+_W$ if and only if there is a path
  from $u$ to $v$ in $G$ with internal vertices in $W$.
  It follows that the value of $k$ found by the algorithm
  is one plus the degree of $v$ in $G^+_W$.
  Therefore, $k=\abs{V}$ if and only if $G^+=K_V$,
  so the algorithm returns the correct decision.
\end{proof}

\begin{lemma}
\label{lem:previous-lemma}
If $\alg{MinimumDegreeOrdering}$ runs in
$O(m^{2-\varepsilon}\Delta^k)$ time (for some $\varepsilon>0$ and $k\ge 0$)
on input $G$ with $m$ edges and
max degree $\Delta$,
then $\alg{CliqueUnion}$ runs in
$O(n^{2-\varepsilon'}d^{k+2})$ time for some $\varepsilon'>0$,
where $n=\abs{V}$.
\end{lemma}

\begin{proof}
For all $i\in[d]$, the graph
$G_i$ is constructed in $O(n\log n)$ time,
has $O(n\log n)$ vertices and edges,
and has maximum degree $O(\log n)$
by \Cref{thm:mindegreefiller}.
Therefore, $G$ has $O(nd \log n)$ vertices and edges
and maximum degree $O(d \log n)$.
We can compute the union of two graphs of size $O(m)$ in $O(m \log m)$ time,
so we can construct $G$ in $O(nd \log^2 n)$ time.
The next step of the algorithm is to run $\alg{MinimumDegreeOrdering}(G)$,
and the running time of this step is
\[
  O\parens*{m^{2-\varepsilon}\Delta^k}
  = O\parens*{(nd\log n)^{2-\varepsilon}(d\log n)^k}
  = O\parens*{n^{2-\varepsilon}d^{k+2}\log^{k+2}n}
  = O\parens*{n^{2-\varepsilon'}d^{k+2}},
\]
for some $\varepsilon'>0$.
%TODO: Check line number.
To determine reachability at line 11 of the algorithm,
it suffices to use a {breadth-first} search that runs in
$O(\abs{E(G)}) = O(nd \log n)$ time.
Therefore, the algorithm runs in $O(n^{2-\varepsilon'}d^{k+2})$ time.
\end{proof}

We conclude with the proof of our improved conditional hardness
result for computing exact minimum degree elimination orderings.
The complementary lower bound in \Cref{cor:hardness} immediately follows.

\begin{proof}[Proof of \Cref{thm:hardness}]
Assume for contradiction that an $O(m^{2-\varepsilon}\Delta^k)$ time algorithm
exists~for~computing~a~minimum~degree~ordering.
By \Cref{lem:reduction-correctness} and \Cref{lem:previous-lemma}, we can
to obtain an $O(n^{2-\varepsilon'}d^{k+2})$ time algorithm
for deciding if ${K_{U_1}\cup K_{U_2}\cup\dots\cup K_{U_d} = K_V}$,
for some $\varepsilon'>0$.
For instances where $d=\Theta(\log^2n)$,
this algorithm runs in time
$O(n^{2-\varepsilon'}\log^{2(k+2)} n)=O(n^{2-\varepsilon''})$,
for some $\varepsilon''>0$.
However, this contradicts SETH by \Cref{lem:clique-union-problem},
so the result follows.
\end{proof}

\section{Conclusion}
\label{sec:conclusion}

% We have made significant progress in characterizing the time complexity for
% computing an exact minimum degree ordering.
% In particular,
% we present a new combinatorial algorithm that runs in $O(nm)$ time.
% This is the first algorithm that improves on the naive $O(n^3)$ algorithm.
% We achieve this using a careful amortized analysis,
% that also leads to strong output-sensitive bounds for the algorithm.

% We also show a matching conditional hardness of $O(m^{2-\varepsilon}\Delta^{k})$,
% for any $\varepsilon > 0$ and $k \ge 0$,
% which affirmatively answers a conjecture~in~\cite{fahrbach2018graph}
% and implies that there are no minimum degree algorithms with
% running time $O(nm^{1-\varepsilon})$ or $O(\sum_{v \in V} \deg(v)^2)$,
% assuming SETH.
% Extending our $U$-filler graph construction to achieve fine-grained hardness
% results for other elimination-based greedy algorithms is
% of independent interest and an exciting future direction of this work.

We have presented a new combinatorial algorithm for computing
an exact minimum degree ordering with an $O(nm)$ worst-case running time.
This is the first algorithm that improves on the naive $O(n^3)$ algorithm.
We achieve this result using a careful amortized analysis,
which also leads to strong output-sensitive bounds for the algorithm.

We also show a matching conditional hardness of $O(m^{2-\varepsilon}\Delta^{k})$,
for any $\varepsilon > 0$ and $k \ge 0$,
which affirmatively answers a conjecture~in~\cite{fahrbach2018graph}
and implies there are no minimum degree algorithms with
running time $O(nm^{1-\varepsilon})$ or $O(\sum_{v \in V} \deg(v)^2)$,
assuming SETH.
Together with the $O(nm)$ algorithm, this nearly characterizes the time complexity
for computing an exact minimum degree ordering.
Extending our $U$-filler graph construction to achieve fine-grained hardness
results for other elimination-based greedy algorithms is
of independent interest and an exciting future direction of this work.

\section*{Acknowledgments}
We thank Richard Peng for many helpful discussions about the ideas in this paper.

\bibliographystyle{alpha}
\bibliography{references}

\newcommand{\etalchar}[1]{$^{#1}$}
\begin{thebibliography}{WAPL14}

\bibitem[ADD96]{amestoy1996approximate}
Patrick~R. Amestoy, Timothy~A. Davis, and Iain~S. Duff.
\newblock An approximate minimum degree ordering algorithm.
\newblock {\em SIAM Journal on Matrix Analysis and Applications},
  17(4):886--905, 1996.

\bibitem[AKR93]{agrawal1993cutting}
Ajit Agrawal, Philip Klein, and Ramamurthy Ravi.
\newblock Cutting down on fill using nested dissection: {P}rovably good
  elimination orderings.
\newblock In {\em Graph Theory and Sparse Matrix Computation}, pages 31--55.
  Springer, 1993.

\bibitem[Ash95]{ashcraft1995compressed}
Cleve Ashcraft.
\newblock Compressed graphs and the minimum degree algorithm.
\newblock {\em SIAM Journal on Scientific Computing}, 16(6):1404--1411, 1995.

\bibitem[BCCR19]{bergman2019minimum}
David Bergman, Carlos~H. Cardonha, Andre~A. Cire, and Arvind~U. Raghunathan.
\newblock On the minimum chordal completion polytope.
\newblock {\em Operations Research}, 67(2):532--547, 2019.

\bibitem[BCK{\etalchar{+}}16]{bliznets2016lower}
Ivan Bliznets, Marek Cygan, Pawel Komosa, Luk{\'a}{\v{s}} Mach, and Micha{\l}
  Pilipczuk.
\newblock Lower bounds for the parameterized complexity of minimum fill-in and
  other completion problems.
\newblock In {\em Proceedings of the Twenty-Seventh Annual ACM-SIAM Symposium
  on Discrete Algorithms}, pages 1132--1151. SIAM, 2016.

\bibitem[BFPP18]{bliznets2018subexponential}
Ivan Bliznets, Fedor~V. Fomin, Marcin Pilipczuk, and Micha{\l} Pilipczuk.
\newblock Subexponential parameterized algorithm for interval completion.
\newblock {\em ACM Transactions on Algorithms (TALG)}, 14(3):1--62, 2018.

\bibitem[CKL{\etalchar{+}}09]{chierichetti2009compressing}
Flavio Chierichetti, Ravi Kumar, Silvio Lattanzi, Michael Mitzenmacher,
  Alessandro Panconesi, and Prabhakar Raghavan.
\newblock On compressing social networks.
\newblock In {\em Proceedings of the 15th ACM SIGKDD international conference
  on Knowledge discovery and data mining}, pages 219--228, 2009.

\bibitem[CLRS09]{cormen2009introduction}
Thomas~H. Cormen, Charles~E. Leiserson, Ronald~L. Rivest, and Clifford Stein.
\newblock {\em Introduction to Algorithms}.
\newblock MIT press, 2009.

\bibitem[CM69]{cuthill1969reducing}
Elizabeth Cuthill and James McKee.
\newblock Reducing the bandwidth of sparse symmetric matrices.
\newblock In {\em Proceedings of the 1969 24th National Conference}, pages
  157--172. ACM, 1969.

\bibitem[CM16]{cao2016chordal}
Yixin Cao and D{\'a}niel Marx.
\newblock Chordal editing is fixed-parameter tractable.
\newblock {\em Algorithmica}, 75(1):118--137, 2016.

\bibitem[CS17]{cao2017minimum}
Yixin Cao and RB~Sandeep.
\newblock Minimum fill-in: Inapproximability and almost tight lower bounds.
\newblock In {\em Proceedings of the Twenty-Eighth Annual ACM-SIAM Symposium on
  Discrete Algorithms}, pages 875--880. SIAM, 2017.

\bibitem[FGP{\etalchar{+}}20]{fahrbach2020faster}
Matthew Fahrbach, Gramoz Goranci, Richard Peng, Sushant Sachdeva, and Chi Wang.
\newblock Faster graph embeddings via coarsening.
\newblock {\em arXiv preprint arXiv:2007.02817}, 2020.

\bibitem[FMP{\etalchar{+}}18]{fahrbach2018graph}
Matthew Fahrbach, Gary~L. Miller, Richard Peng, Saurabh Sawlani, Junxing Wang,
  and Shen~Chen Xu.
\newblock Graph sketching against adaptive adversaries applied to the minimum
  degree algorithm.
\newblock In {\em 2018 IEEE 59th Annual Symposium on Foundations of Computer
  Science (FOCS)}, pages 101--112. IEEE, 2018.

\bibitem[FV13]{fomin2013subexponential}
Fedor~V. Fomin and Yngve Villanger.
\newblock Subexponential parameterized algorithm for minimum fill-in.
\newblock {\em SIAM Journal on Computing}, 42(6):2197--2216, 2013.

\bibitem[Geo73]{george1973nested}
Alan George.
\newblock Nested dissection of a regular finite element mesh.
\newblock {\em SIAM Journal on Numerical Analysis}, 10(2):345--363, 1973.

\bibitem[GL79]{george1979quotient}
Alan George and Joseph W.~H. Liu.
\newblock A quotient graph model for symmetric factorization.
\newblock In {\em Sparse Matrix Proceedings}, pages 154--175, 1979.

\bibitem[GL80]{george1980minimal}
Alan George and Joseph W.~H. Liu.
\newblock A minimal storage implementation of the minimum degree algorithm.
\newblock {\em SIAM Journal on Numerical Analysis}, 17(2):282--299, 1980.

\bibitem[GL81]{george1981computer}
Alan George and Joseph~W. Liu.
\newblock {\em Computer Solution of Large Sparse Positive Definite}.
\newblock Prentice Hall Professional Technical Reference, 1981.

\bibitem[GL89]{george1989evolution}
Alan George and Joseph W.~H. Liu.
\newblock The evolution of the minimum degree ordering algorithm.
\newblock {\em {SIAM} Review}, 31(1):1--19, 1989.

\bibitem[GNP94]{gilbert1994efficient}
John~R. Gilbert, Esmond~G. Ng, and Barry~W. Peyton.
\newblock An efficient algorithm to compute row and column counts for sparse
  cholesky factorization.
\newblock {\em SIAM Journal on Matrix Analysis and Applications},
  15(4):1075--1091, 1994.

\bibitem[GVL13]{golub2013matrix}
Gene~H. Golub and Charles~F. Van~Loan.
\newblock {\em Matrix Computations}.
\newblock The Johns Hopkins University Press, fourth edition, 2013.

\bibitem[HEKP01]{heggernes2001computational}
Pinar Heggernes, S.~C. Eisestat, Gary Kumfert, and Alex Pothen.
\newblock The computational complexity of the minimum degree algorithm.
\newblock Technical report, Institute for Computer Applications in Science and
  Engineering, 2001.

\bibitem[KST99]{kaplan1999tractability}
Haim Kaplan, Ron Shamir, and Robert~E. Tarjan.
\newblock Tractability of parameterized completion problems on chordal,
  strongly chordal, and proper interval graphs.
\newblock {\em SIAM Journal on Computing}, 28(5):1906--1922, 1999.

\bibitem[Liu85]{liu1985modification}
Joseph W.~H. Liu.
\newblock Modification of the minimum-degree algorithm by multiple elimination.
\newblock {\em ACM Transactions on Mathematical Software (TOMS)},
  11(2):141--153, 1985.

\bibitem[LRT79]{lipton1979generalized}
Richard~J. Lipton, Donald~J. Rose, and Robert~E. Tarjan.
\newblock Generalized nested dissection.
\newblock {\em SIAM Journal on Numerical Analysis}, 16(2):346--358, 1979.

\bibitem[Mar57]{markowitz1957elimination}
Harry~M. Markowitz.
\newblock The elimination form of the inverse and its application to linear
  programming.
\newblock {\em Management Science}, 3(3):255--269, 1957.

\bibitem[MB83]{matula1983smallest}
David~W. Matula and Leland~L. Beck.
\newblock Smallest-last ordering and clustering and graph coloring algorithms.
\newblock {\em Journal of the ACM (JACM)}, 30(3):417--427, 1983.

\bibitem[NSS00]{natanzon2000polynomial}
Assaf Natanzon, Ron Shamir, and Roded Sharan.
\newblock A polynomial approximation algorithm for the minimum fill-in problem.
\newblock {\em SIAM Journal on Computing}, 30(4):1067--1079, 2000.

\bibitem[OSS20]{ost2020engineering}
Wolfgang Ost, Christian Schulz, and Darren Strash.
\newblock Engineering data reduction for nested dissection.
\newblock {\em arXiv preprint arXiv:2004.11315}, 2020.

\bibitem[Ros72]{rose1972graph}
Donald~J. Rose.
\newblock A graph-theoretic study of the numerical solution of sparse positive
  definite systems of linear equations.
\newblock In {\em Graph Theory and Computing}, pages 183--217. Elsevier, 1972.

\bibitem[Tar76]{tarjan1976graph}
Robert~E. Tarjan.
\newblock Graph theory and gaussian elimination.
\newblock In {\em Sparse Matrix Computations}, pages 3--22. Elsevier, 1976.

\bibitem[WAPL14]{wu2014inapproximability}
Yu~Wu, Per Austrin, Toniann Pitassi, and David Liu.
\newblock Inapproximability of treewidth and related problems.
\newblock {\em Journal of Artificial Intelligence Research}, 49:569--600, 2014.

\bibitem[Wil05]{williams2005new}
Ryan Williams.
\newblock A new algorithm for optimal 2-constraint satisfaction and its
  implications.
\newblock {\em Theoretical Computer Science}, 348(2-3):357--365, 2005.

\bibitem[Wil18]{williams2018some}
Virginia~Vassilevska Williams.
\newblock On some fine-grained questions in algorithms and complexity.
\newblock In {\em Proceedings of the ICM}, 2018.

\bibitem[Yan81]{yannakakis1981computing}
Mihalis Yannakakis.
\newblock Computing the minimum fill-in is {NP}-complete.
\newblock {\em SIAM Journal on Algebraic Discrete Methods}, 2(1):77--79, 1981.

\end{thebibliography}

\newpage
\appendix

\end{document}